\documentclass[a4paper,UKenglish]{lipics}

\usepackage{microtype}

\usepackage{amsmath,amssymb,amsthm,amsfonts}

\usepackage{tikz}
\usetikzlibrary{calc,automata,arrows,chains,matrix,positioning,scopes}

\theoremstyle{plain}

\newcommand{\bh}{\operatorname{bh}}
\newcommand{\LL}{\sim_{\ell\ell}}
\newcommand{\LR}{\sim_{\ell r}}
\newcommand{\RL}{\sim_{r\ell}}
\newcommand{\RR}{\sim_{rr}}
\newcommand{\FO}{\mathbf{FO}}
\newcommand{\MSO}{\mathbf{MSO}}

\newcommand{\cA}{\mathcal{A}}
\newcommand{\cB}{\mathcal{B}}
\newcommand{\cC}{\mathcal{C}}

\newcommand{\trans}[1]{\mathchoice{\xrightarrow{#1}}{\xrightarrow{\smash{\lower1pt\hbox{$\scriptstyle #1$}}}}{\text{Error}}{\text{Error}}}
\newcommand{\isp}[0]{\hphantom{a}}

\renewcommand{\leq}{\leqslant}

\newcommand{\tvi}{\vrule height 12pt depth 5 pt width 0 pt}
\newsavebox{\Jclass}
\sbox{\Jclass}{
  \begin{tabular}{|c|c|}
  \hline
  \tvi $^*aba$ & $^*abb$ \\
  \hline
  \tvi $^*bba$ & $^*bb$ \\
  \hline
  \end{tabular}
  }

\title{Aperiodic two-way transducers and $\FO$-transductions}

\author[1]{Olivier Carton}
\author[2,3]{Luc Dartois}
\affil[1]{LIAFA, Universit\'e Paris Diderot \\
  \texttt{Olivier.Carton@liafa.univ-paris-diderot.fr}}
\affil[2]{LIF, UMR7279 Aix-Marseille Université \& CNRS \\
  \texttt{luc.dartois@lif.univ-mrs.fr}}
\affil[3]{Centrale Marseille}
\authorrunning{O. Carton and L. Dartois} 

\Copyright{Olivier Carton and Luc Dartois}


\serieslogo{}
\volumeinfo
  {Billy Editor and Bill Editors}
  {2}
  {Conference title on which this volume is based on}
  {1}
  {1}
  {1}
\EventShortName{}
\DOI{10.4230/LIPIcs.xxx.yyy.p}

\begin{document}

\maketitle

\begin{abstract}
  Deterministic two-way transducers on finite words have been shown by
  Engelfriet and Hoogeboom to have the same expressive power as
  $\MSO$-transductions.  We introduce a notion of aperiodicity for these
  transducers and we show that aperiodic transducers correspond exactly
  to $\FO$-transductions.  This lifts to transducers the classical equivalence
  for languages between $\FO$-definability, recognition by aperiodic
  monoids and acceptance by counter-free automata.
\end{abstract}

\section{Introduction} \label{sec:introduction}

The regularity of a language of finite words is a central notion in
theoretical computer science.  Combining several seminal results, it is
equivalent whether a language is
\begin{enumerate}
\item[a)] accepted by a (non-)deterministic one-way or two-way automaton
  \cite{RabinScott59} and \cite{She59},
\item[b)] described by a regular expression \cite{Kleene56},
\item[c)] defined in (Existential) Monadic Second Order ($\MSO$) logic
  \cite{Buchi60},
\item[d)] the preimage by a morphism into a finite monoid \cite{Nerode59}.
\end{enumerate}
Since then, the characterization of fragments of $\MSO$ has been a very
successful story.  Using this equivalence between different formalisms,
several fragments of $\MSO$ have been characterized by algebraic means and
shown to be decidable. Combining results of Sch\"utzenberger
\cite{Schutzenberger65} and of McNaughton and Papert \cite{MP71} yields,
for instance, that a language of finite words is First Order ($\FO$) definable
if and only if all the groups contained in its syntactic monoid are trivial
(aperiodic).  From the results of Sch\"utzenberger \cite{Schutzenberger76}
and others \cite{TherienWilke98}, it is also known that a language is First
Order definable with two variables ($\FO^2$) if and only if its syntactic
monoid belongs to the class $\mathbf{DA}$ which is easily decidable.

Automata can be equipped with output to make them compute functions and
relations.  They are then called transducers.  Note then that all variants
are no longer equivalent as they are as acceptors.  Deterministic
transducers compute a subclass of rational functions called sequential
functions \cite{Choffrut77}.  Two-way transducers are also more powerful
than one-way transducers (see Example~\ref{exa:transducer}).  The study of
transducers has many applications.  Transducers are used to model coding
schemes (compression schemes, convolutional coding schemes, coding schemes
for constrained channels, for instance).  They are also widely used in
computer arithmetic \cite{Frougny99}, natural language processing
\cite{RocheSchabes97} and programs analysis \cite{CohenCollard98}.

The equivalence between automata and $\MSO$ has been first lifted to
transducers and the functions they realize by Engelfriet and Hoogeboom
\cite{EH01}.  They show that a function from words to words can be realized
by a deterministic two-way transducer if and only it is a $\MSO$-transduction.
First, this result deals surprisingly with two-way transducers rather than
one-way transducers which are much simpler.  Second, the $\MSO$-definability
used for automata is replaced by $\MSO$ graphs transductions defined by
Courcelle \cite{Cou94}.  A $\MSO$-transduction is a function where the output
graph is defined as a $\MSO$-interpretation into a fixed number of copies of
the input graph.  In the result of Engelfriet and Hoogeboom, words are seen
as linear graphs whose vertices carry the symbols.

\subsection*{Contribution}

In this paper, we combine the approach of Engelfriet and Hoogeboom with the
one of Sch\"utzenberger, McNaughton and Papert.  We introduce a notion of
aperiodicity for two-way transducers and we show that it corresponds to
$\FO$-transductions.  By $\FO$-transduction, we mean $\MSO$-transduction
where the interpretation is done through $\FO$-formulas.  The definition of
aperiodicity is achieved by associating a transition monoid with each
two-way transducer.  The construction of this algebraic object is already
implicit in the literature \cite{She59,Pec85,Bir89}.  In order to obtain
our result, we have considered a different logical signature for
transductions from the one used in \cite{EH01}.  In \cite{EH01}, the
signature contains the symbol predicates to check symbols carried by
vertices and the edge predicate of the graph.  Since words are viewed as
linear graphs, this is the same as the signature with the successor
relation on words.  In our result, the signature contains the symbol
predicates and the order (of the linear graph).  This is equivalent for
$\MSO$-transductions since the order can easily be defined with the
successor by a $\MSO$-formula.  This is however not equivalent any more for
$\FO$-transductions that we consider.  With this signature, the definition
of $\FO$-transduction requires that the order on the output word can be
defined by a $\FO$-formula.  The change in the signature is necessary to
obtain the result.

\subsection*{Related work}

The aperiodic rational functions, that is, functions realized by a one-way
transducer with an aperiodic transition monoid have already been
characterized in \cite{ReutenauerSchutz95}.  This characterization is not
based on logic but rather on the inverse images of aperiodic languages.

The notion of aperiodic two-way transducer was already defined and studied
in~\cite{MSAl00}, although their model defined length-preserving functions
and the transducers had both their reading and writing heads moving
two-way.  The assumption that the function is length preserving makes the
relation between the input and the output easier to handle.

Recently, Bojanczyk, in~\cite{Boj14}, also characterized first-order
definable transducers for machines using a finer but more demanding
semantic, the so-called origin semantic.

In~\cite{AC10}, Alur and \v Cern\'y defined the streaming string
transducers, a one-way deterministic model equivalent to deterministic
two-way transducers and $\MSO$ transductions.  More recently, Filiot,
Krishna and Trivedi proposed in~\cite{FKT14} a definition of transition
monoid for this model. They also proved that aperiodic and $1$-bounded
streaming string transducers have the same expressive power as $\FO$
transductions, which is one of the models considered by our main result.

\subsection*{Structure}

The paper is organized as follows.  Definitions of two-way transducers and
FO-transductions are provided in Section~\ref{sec:definitions}.  The
construction of the transition monoid associated with a transducer is given
there.  The main result is stated in Section~\ref{sec:main}.
Section~\ref{sec:composition} focuses on one aspect of the stability by
composition of functions realized by aperiodic two-way transducers. It is
one of the main ingredients used in the proof of the main result.  The
proof itself is sketched in Sections \ref{sec:2WtoFOT}
and~\ref{sec:FOTtoA2W}.

\section{Definitions} \label{sec:definitions}

In this section, we present the different models that will be used
throughout the article.

\subsection{Two-way transducers}\label{Subsection:2w}

A transducer is an automaton equipped with outputs.  While an input word is
processed along a run by the transducer, each used transition outputs some
word.  All these output words are concatenated to form the output of the run.
The automaton might be one-way or two-way but we mainly consider two-way
transducers in this paper.  When the transducer is non-deterministic, there
might be several runs and therefore several output words for a single input
word.  All two-way transducers considered in this paper are deterministic.
For each input word, there is then at most one valid run and one output
word.  The partial function which maps each input word to the corresponding
output word is said to be \emph{realized} by the transducer.  The automaton
obtained by forgetting the outputs is called the \emph{input automaton} of
the transducer.

A two-way transducer is a very restricted variant of a Turing machine with
an input and an output tape.  First, the input tape is read-only. Second,
the output tape is write-only and the head on this tape only moves forwards.
Written symbols on this tape cannot be over-written later by other symbols.

\begin{figure}
  \begin{center}
   \begin{minipage}[c]{0.4\linewidth}
	  \begin{tikzpicture}[scale=0.80, initial text=,inner sep=0pt]
  \node[state,initial,accepting where=left, initial where=left] (q) at (0,0) {$1$};
  \node[state] (r) at (4,0) {$2$};
  \node[state, accepting,accepting where=left] (s) at (2,-2) {$3$};
  \path[->] (q) edge [loop above]   node   {\small{$a|a,+1$}}();
  \node () at (-0.04,1.84) { \small{${\vdash}|\epsilon,+1$}};
  \path[->] (q) edge [above] node {\small{$b|\epsilon,-1$}}(r);
  \node () at (1.94,0.70) { \small{${\dashv}|\epsilon,-1$}};
  \path[->] (r) edge [loop above]   node   {\small{$a|b,-1$}}();
  \path[->] (r) edge [below right] node {\small{$b|\epsilon,+1$}} (s);
  \node () at (3.57,-1.67) {\small{${\vdash}|\epsilon,+1$}};
  \path[->] (s) edge [loop below] node {\small{$a|\epsilon,+1$}} ();
  \path[->] (s) edge [below left] node {\small{$b|\epsilon,+1$}} (q);
		  \end{tikzpicture} 
 	 \end{minipage}\hfill
   \begin{minipage}[c]{0.6\linewidth}
  \begin{tikzpicture}[scale=0.7]
  \foreach \u/\utext in {0/{\vdash},1/a,2/a,3/b,4/a,5/b,6/b,7/{\dashv}}
          \node (u\u) at (\u+7,1) {$\utext$};
  \foreach \p in {0,1,2,3,4,5,6,7}
          \node (p\p) at (\p+7,0) {$1$};
  \foreach \s/\setat in {0/2,1/2,2/2,3/3,4/2,5/3,6/3,7/3}
          \node (s\s) at (\s+7,-1) {$\setat$};
  \foreach \t/\tetat in {1/3,2/3,3/2,4/3,5/2,6/2}
          \node (t\t) at (\t+7,-2) {$\tetat$};
  \node (debut) at (6,0) {};
  \node (fin) at (15,-1) {};
  \foreach \d/\a in {	debut/p0, p0/p1,  
                                          p3/s2,s0/t1,t1/t2,t2/s3, s3/p4,
                                          p5/s4,t3/t4, t4/s5,
                                          s5/p6,p6/t5, t5/s6, s6/p7, 
                                          p7/t6,t6/s7,s7/fin}
          \path[->] (\d) edge node {}(\a);
  \foreach \d/\a/\et in {	p1/p2/a, p2/p3/a, 
                                           s2/s1/b, s1/s0/b,
                                          p4/p5/a,s4/t3/b}
          \path[->] (\d) edge node [above] {{\small $\et$}}(\a);	
  \node (u) at (5.5, 1) {input $w$:};
  \node () at (5,-1) {run};
  \node () at (6.8,-3) {{output: $f(w)=aabbab$}};
  \end{tikzpicture}
  \end{minipage}
  \end{center}
  \caption{A transducer and its run over $w = aababb$}
  \label{fig:transducer}
\end{figure}

\begin{example}\label{exa:transducer}
  Let $A$ be the alphabet $\{a, b\}$.  Let us consider, as a running
  example, the function $f: A^*\to A^*$ which maps each word $w =
  a^{k_0}ba^{k_1}\cdots ba^{k_n}$ to the word $f(w) =
  a^{k_0}b^{k_0}a^{k_1}b^{k_1} \cdots a^{k_n}b^{k_n}$ obtained by adding
  after each block of consecutive~$a$ a block of consecutive~$b$ of the
  same length.  Since each word $w$ over~$A$ can be uniquely written $w =
  a^{k_0}ba^{k_1} \cdots ba^{k_n}$ with some $k_i$ being possibly equal to
  zero, the function~$f$ is well defined.  The word $w = aababb =
  a^2ba^1ba^0ba^0$ is mapped to $f(w) = a^2b^2a^1b^1a^0b^0a^0b^0 = aabbab$.

  This function is realized by the transducer depicted in
  Figure~\ref{fig:transducer}.  This transducer proceeds as follows to
  compute $f(w)$ from the input word~$w$.  While being in state~$1$ and
  moving forwards, it copies a block of consecutive~$a$ to the output.
  While in state~$2$ and moving backwards, the corresponding block of~$b$
  is written to the output.  While being in state~$3$, the transducer moves
  forwards writing nothing until it reaches the next block of
  consecutive~$a$.  Note that this function cannot be realized by a
  one-way transducer.
\end{example}

Formally, a two-way transducer is defined as follows:

\begin{definition}[Two-way transducer]
  A \emph{(deterministic) two-way transducer} $\mathcal{A}$ is a tuple
  $\mathcal{A}=(Q,A,B,\delta,\gamma, q_0,F))$ defined as follows:
\begin{itemize} \itemsep0cm
\item $Q$ is a finite \emph{state set}.
\item $A$ and $B$ are the \emph{input} and \emph{output alphabet}.
\item $\delta: Q\times (A\uplus\{{\vdash},{\dashv}\}) \to Q\times\{-1,0,+1\}$ is the \emph{transition function}. Contrary to the one-way machines, the transition function also outputs an integer, corresponding to the move of the reading head.
The alphabet is enriched with two new symbols $\vdash$ and $\dashv$, which are endmarkers that are added respectively at the beginning and the end of the input word, such that for all $q\in Q$, we have ${\delta(q,\vdash)\in Q\times\{0,+1\}}$ and $\delta(q,\dashv)\in Q\times\{-1,0\}$.
\item $\gamma: Q\times (A\uplus\{{\vdash},{\dashv}\}) \to B^*$ is the \emph{production function}.
\item $q_0\in Q$ is the \emph{initial state}.
\item $F \subseteq Q$ is the set of final states.
\end{itemize}
\end{definition}

The transducer~$\mathcal{A}$ processes 
finite words over $A$. If at state
$p$ the symbol $a$ is processed and $\delta(p,a) = (q,d)$, then
$\mathcal{A}$ moves to state $q$, moves the reading head to the left or
right depending on $d$, and outputs $\gamma(p,a)$.

Let $w = a_1 \cdots a_n$ be a fixed finite word over $A$ and $a_0 =
{\vdash}$ and $a_{n+1} = {\dashv}$.  Whenever $\delta(p,a_m) = (q,d)$ and
$\gamma(p,a_m) = v$, we write $(p,m) \trans{\isp|v} (q,n)$ where $n =
m+d$. We do not write the input over the arrow because it is always the
symbol below the reading head, namely, $a_m$.  In this notation, the pairs
represent the current configuration of a machine with the current state and
the current position of the input head.
A \emph{run} of the transducer over $w$ is a finite sequence of consecutive transitions
\begin{displaymath}
(p_0,m_0) \trans{\isp|v_1} (p_1,m_1)  \cdots (p_{n-1},m_{n-1}) \trans{\isp|v_n} (p_n,m_n)
\end{displaymath}
and we write $(p_0,m_0) \trans{\isp|v} (p_n,m_n)$ where $v = v_1 v_2 \cdots
v_n$.  We also refer to finite runs over words $w$ when all positions $m_i$
in the run but the last are between $1$ and~$|w|$.  The last position~$m_n$
is allowed to be between $0$ and $|w|+1$.  It is $0$ if the run leaves
$w$ on the left end and it is $|w|+1$ if it leaves $|w|$ on the right end.

A run $(p_0,m_0) \trans{\isp|v} (p_n,m_n)$ over a marked word ${\vdash}
u{\dashv}$ is \emph{accepting} if it starts at the first position in the
initial state and ends on the right endmarker $\dashv$ in a final state.
Then $v$ is the \emph{image} of $u$ by $\mathcal{A}$, denoted
$\mathcal{A}(u)=v$.

\subsection{Transition monoid}

In order to define a notion of aperiodicity for a transducer, we associate with
each two-way automaton a monoid called its \emph{transition monoid}.  A
transducer is then called \emph{aperiodic} if the transition monoid of its
input automaton is aperiodic.  Let us recall that a monoid is called
aperiodic if it contains no trivial group \cite{Almeida94}. 
Equivalently, a monoid $M$ is aperiodic if there exists a smallest integer $n$, called the \emph{aperiodicity index}, such that
for any element $x$ of $M$, we have $x^n=x^{n+1}$.
 Note first
that the transition monoid of a transducer is the transition monoid of its
input automaton and does not depend of its outputs.  Note also that our
definition is sound for either deterministic or non-deterministic
automata/transducers although we only use it for deterministic ones. Lastly, remark that it extends naturally the notion of transition monoid for one-way automata.

The transition monoid is, as usual, obtained by quotienting the free
monoid~$A^*$ by a congruence which captures the fact that two words have
the same \emph{behavior} in the automaton.  In an one-way
automaton~$\mathcal{A}$, the behavior of a word~$w$ is the set of pairs
$(p,q)$ of states such that there exists a run from~$p$ to~$q$
in~$\mathcal{A}$.  Two words are then considered equivalent if their
respective behaviors contain the same pairs of states.  In a two-way
automaton, the behavior of a word is also characterized by the runs it
contains but since the reading head can move both ways, the behavior is
split into four behaviors called left-to-left, left-to-right, right-to-left
and right-to-right behaviors.  We only define the left-to-left
behavior~$\bh_{\ell\ell}(w)$ of a word~$w$.  The three other behaviors
$\bh_{\ell r}(w)$, $\bh_{r\ell}(w)$ and~$\bh_{rr}(w)$ are defined analogously.

Let $\mathcal{A}$ be a two-way automaton. The \emph{left-to-left behavior}
$\bh_{\ell\ell}(w)$ of~$w$ in~$\mathcal{A}$ is the set of pairs $(p,q)$ such that
there exists a run which starts at the first position of~$w$ in state~$p$
and leaves $w$ on the left end in state~$q$ (see Figure~\ref{fig:ll-path}).

Before defining the transition monoid, we illustrate the notion of
behavior on the transducer depicted in Figure~\ref{fig:transducer}.
\begin{example}
  Consider the transducer depicted in Figure~\ref{fig:transducer} and the
  word $w = aab$.  From the run depicted in Figure~\ref{fig:transducer}, it
  can be inferred that
  \begin{alignat*}{2}
    \bh_{\ell\ell}(w) & = \{ (1,2), (2,2) \} & \qquad \bh_{r\ell}(w) & = \{ (1,2) \} \\
    \bh_{\ell r}(w) & = \{ (3,1) \} &  \qquad \bh_{rr}(w) & = \{ (2,3), (3,1) \}.
  \end{alignat*}
\end{example}

\begin{definition}[Transition monoid]
  Let $\mathcal{A}=(Q,A,\delta,q_0,F)$ be a two-way automaton.  The
  transition monoid of~$\mathcal{A}$ is $A^*/\!\!\sim_{\mathcal{A}}$ where
  $\sim_{\mathcal{A}}$ is the conjunction of the four relations $\LL$,
  $\LR$, $\RL$ and $\RR$ defined for any words $w$, $w'$ of $A^*$ as follows :
  \begin{itemize}
  \item $w \LL w'$ if $\bh_{\ell\ell}(w) = \bh_{\ell\ell}(w')$.
  \item $w \LR w'$ if $\bh_{\ell r}(w) = \bh_{\ell r}(w')$.
  \item $w \RL w'$ if $\bh_{r\ell}(w) = \bh_{r\ell}(w')$.
  \item $w \RR w'$ if $\bh_{rr}(w) = \bh_{rr}(w')$.
  \end{itemize}
  The neutral element of this monoid is the class of the empty
  word~$\epsilon$, whose behaviors $bh_{xy}(\epsilon)$ is the identity
  function if $x\neq y$, and is the empty relation otherwise.
\end{definition}

\begin{figure}[htb]
  \begin{center}
  \begin{tikzpicture}
  \draw[|-|] (0,1) -- (3,1);
  \node (u) at (1.5,1.5) {$w$};
  \node (q) at (0,0.6) {$p$};
  \node (p) at (-0.4,0) {$q$};
  \draw plot [smooth] coordinates { (0.1,0.6) (2,0.6) (1,0.4) (1.4,0.2) (0,0)};
  \draw[->] (0,0) -- (p);
  \end{tikzpicture}
  \end{center}
  \caption{A left-to-left behavior $(p,q)$ of a word $w$.}
  \label{fig:ll-path}
\end{figure}
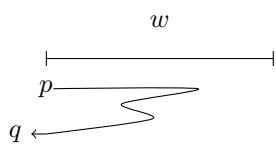

These relations are not new and were already evoked in~\cite{Pec85,Bir89}
for example.  Moreover, the left-to-left behavior was already introduced
in~\cite{She59} to prove the equivalence between one-way and two-way
automata.  

For a deterministic two-way
automaton, the four behaviors
$\bh_{\ell\ell}(w)$, $\bh_{\ell r}(w)$, $\bh_{r\ell}(w)$ and~$\bh_{rr}(w)$ are partial functions.
In the non-deterministic case, these four relations are not functions but relations
over the state set~$Q$ because there might exist several runs with the same
starting state and different ending states. 
 Furthermore, for deterministic automaton,
 the domains of the functions $\bh_{\ell\ell}(w)$ and $\bh_{\ell r}(w)$ (resp. $\bh_{r\ell}(w)$ and~$\bh_{rr}(w)$) are disjoint,
 since
there is a unique run starting in state~$p$ at the first (resp. last) position of~$w$.
Thus a run starting at the first (resp. last) position leaves $w$ either on the left or the right.
  For a deterministic two-way automaton, the four behaviors
$\bh_{\ell\ell}(w)$, $\bh_{\ell r}(w)$, $\bh_{r\ell}(w)$ and~$\bh_{rr}(w)$ can be seen
as a single partial function $f_w$ from $Q \times \{\ell,r\}$ to~$Q \times \{\ell,r\}$
where $f_w(p,x) = (q,y)$ whenever $(p,q) \in \bh_{xy}(w)$ for any $x,y \in
\{\ell,r\}$. 

\begin{lemma}\label{LemmaCongruence}
  Let $\mathcal{A}$ be a two-way transducer.  Then the relation
  $\sim_{\mathcal{A}}$ is a congruence of finite index.
\end{lemma}

It is pure routine to check that $\sim_{\mathcal{A}}$ is indeed a
congruence.  It is of finite index since each of the four relations $\LL$,
$\LR$, $\RL$ and~$\RR$ has at most $2^{|Q|^2}$ classes.  Note that the
composition of the behaviors is not as straightforward as in the case of
one-way automata, the four relations being intertwined.  For example, the
composition law of the $\bh_{\ell r}$ relation is given by the equality
$\bh_{\ell r}(uv) = \bh_{\ell
  r}(u)\big(\bh_{\ell\ell}(v)\bh_{rr}(u)\big)^*\bh_{\ell r}(v)$ which
follows from the decomposition of a run in~$uv$.

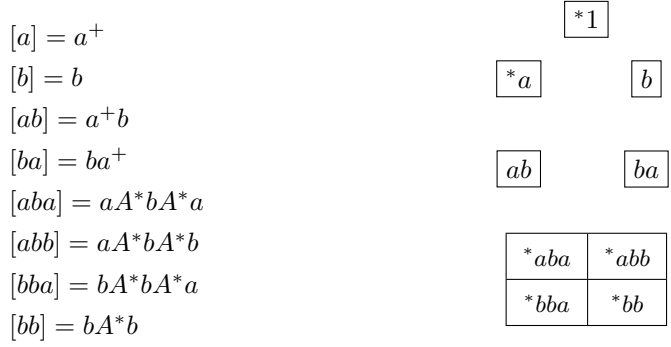
\begin{figure}
  \begin{center}
    \begin{minipage}[c]{0.5\linewidth}
      \begin{displaymath}
        \begin{array}{l}
          [a]=a^+    \\[0.6ex] 
          [b]= b     \\[0.6ex]
          [ab]=a^+b  \\[0.6ex] 
          [ba]=ba^+  \\[0.6ex]	
          [aba]=aA^*bA^*a \\[0.6ex] 
          [abb]=aA^*bA^*b \\[0.6ex]
          [bba]=bA^*bA^*a \\[0.6ex] 
          [bb]=bA^*b
        \end{array}
      \end{displaymath}
    \end{minipage}\hfill
    \begin{minipage}[c]{.5\linewidth}
    \begin{tikzpicture}
    \node[draw] (Neutre) {$^*1$};
    \node[draw,below left=1.2em of Neutre] (a) {$^*a$};
    \node[draw, below right=1.2em of Neutre] (b) {$b$};
    \node[draw, below= 2em of a] (ab) {$ab$};
    \node[draw, below= 2em of b] (ba) {$ba$};
    \node[below =7em of Neutre] (Jclass) {\usebox{\Jclass}};
    \end{tikzpicture}
    \end{minipage}
  \end{center}
  \caption{The equivalence classes of the transition monoid and its $\mathcal{D}$-class representation}
  \label{fig:monoid}
\end{figure}

\begin{example} \label{exa:monoid} 
  We illustrate the notion of a transition monoid by giving the one of the
  transducer depicted in Figure~\ref{fig:transducer}.  We have omitted all
  words containing one of the two endmarkers since these words cannot
  contribute to a group.  The eight classes of the
  congruence~$\sim_{\mathcal{A}}$ for the remaining words are given in
  Figure~\ref{fig:monoid} on the left.  The $\mathcal{D}$-class representation
  of this monoid is also given for the aware reader on the right.  It can
  be checked that this monoid is aperiodic.  The transducer of
  Figure~\ref{fig:transducer} is then aperiodic.
\end{example}

\subsection{$\FO$ graph transductions}

The $\MSO$-transductions defined by Courcelle \cite{Cou94} are a variant of
the classical logical interpretation of a relational structure into another
one.  Let us recall that a relational structure~$S$ has a
$\mathcal{L}$-interpretation, for some logic~$\mathcal{L}$, into a
structure~$T$ if it has an isomorphic copy in~$T$ defined by
$\mathcal{L}$-formulas.  More precisely, this means that there exists a
$\mathcal{L}$-formula~$\varphi_S$ with one first-order free variable and a
one-to-one correspondence~$f$ between the domain of~$S$ and the subset~$T'$
of elements of~$T$ satisfying~$\varphi_S$.  Furthermore, for each
relation~$R$ of~$S$ with arity~$r$, there exists a
$\mathcal{L}$-formula~$\varphi_R$ with $r$ first-order free variables such
that $R$ is isomorphic via~$f$ to the $r$-tuples of~$T'$
satisfying~$\varphi_R$.

A $\MSO$-transduction defines for each input structure a new structure
obtained by $\MSO$-interpretation into a fixed number of copies of the
input structure. In this case, the relations are the letter predicates and
the successor relation, which are of arity one and two respectively. To fit
into this framework, words are viewed as linear graphs.  Each word $w =
a_1\cdots a_n$ is viewed as a linear graph with $n$ vertices carrying the
symbols $a_1,\ldots,a_n$.  Linear means here that if the vertex set is
$\{1,2,\ldots,n\}$, the edge set is $\{ (k,k+1) : 1 \le k \le n-1\}$.

When restricted to linear graphs, the $\MSO$-transductions has been proved
to have the same expressive power as two-way transducers \cite{EH01}.  We
are interested in this article in $\FO$ graph transductions, the
restriction to first order formulas.  Since we consider transductions whose
domain is not the set of all graphs, there is an additional closed formula
$\varphi_{dom}$ which determines whether the given graph is in the domain
of the transduction.

\begin{figure}
\begin{tikzpicture}[scale=0.6]

\foreach \u/\utext in {1/a,3/a,5/b,7/a,9/b,11/b}
	\node[minimum width=2em,draw,shape=circle] (u\u) at (\u,1) {$\utext$};

\foreach \a/\b in {1/3,3/5,5/7,7/9,9/11}
\path[->] (u\a) edge   node   {}(u\b);

\node (u) at (-0.3,1) {$u$ :};
\node (tu) at (-0.5,-1.25) {$T(u)$:};
\node (c1) at (12.5, -0.5) {\footnotesize{copy $1$}};
\node (c2) at (12.5, -2) {\footnotesize{copy $2$}};

\foreach \u/\utext in {1/a,3/a,7/a}{
	\node[minimum width=2em,draw,shape=circle] (t1\u) at (\u,-0.5) {$\utext$};
	\node[minimum width=2em,draw,shape=circle] (t2\u) at (\u,-2) {$b$};
}
\path[->] (t11) edge node {} (t13);
\path[->] (t13) edge node {} (t21);
\path[->] (t21) edge node {} (t23);
\path[->] (t23) edge node {} (t17);
\path[->] (t17) edge node {} (t27);

\foreach \u/\utext in {5/b,9/b,11/b}{
	\node[minimum width=2em,draw,shape=circle,color=gray!70] (t1\u) at (\u,-0.5) { };
	\node[minimum width=2em,draw,shape=circle,color=gray!70] (t1\u) at (\u,-2) { };
	 }
\end{tikzpicture}
\caption{The linear graph of $u=aababb$ and the output structure of $T$ over $u$.}\label{Ex-onLogic}
\end{figure}
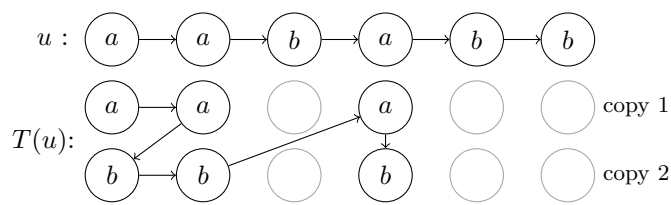

Before giving the formal definition, we give below an example of a
$\FO$-transduction. Note that when considering $\FO$ transductions, the successor relation is replaced by the order relation.

\begin{example}\label{exa:FO-tranduction}
We give here a $\FO$ graph transduction that realizes the function $f$
introduced in Example~\ref{exa:transducer}.
So let $T=(A,A,\varphi_{dom},C,\varphi
_{pos},\varphi_\leq)$ be the $\FO$ graph transduction defined as follows :
\begin{itemize}
\item $A=\{a,b\}$ is both the input and output alphabet,
\item $C=\{1,2\}$,
\item $\varphi_{dom}$ is a $\FO$ formula stating that the input is a linear graph,
\item $\varphi_a^1(x)=\varphi_b^2(x)=\mathbf{a}(x)$, the other position formulas being set as $false$,
\item the order formulas are defined now :
		\begin{itemize}
		\item $\varphi_\leq^{i,i}(x,y)=x\leq y$ for $i=1,2$,
		\item $\varphi_\leq^{1,2}(x,y)=x\leq y \vee (\forall z\ y\leq z\leq x \to \mathbf{a}(z))$,
		\item $\varphi_\leq^{2,1}(x,y)=\exists z\ x\leq z\leq y \wedge \mathbf{b}(z)$.
		\end{itemize}
\end{itemize}
\end{example}

\begin{definition}
A $\FO$-graph transduction is a tuple $T=(A,B,\varphi_{dom},C,\varphi_{pos},\varphi_\leq)$ defined as follows:
\begin{itemize}
\item $A$ is the \emph{input alphabet}.
\item $B$ is the \emph{output alphabet}.
\item $\varphi_{dom}$ is the \emph{domain formula}. A graph is accepted as input if it satisfies the domain formula.
\item $C$ is a finite set, denoting the copies of the input that can exist in the output.
\item $\varphi_{pos}$ is a set of formulas with one free variable $\varphi_b^c(x)$, for $b\in B$ and $c\in C$. Given $c$, the formulas $\varphi_b^c(x)$, for $b\in B$, are mutually exclusive.
The $c$ copy of a node $i$ is labelled by $b$ if, and only if, the formula $\varphi_b^c(x/i)$ is true.
\item $\varphi_\leq$ is a set of formulas with two free variables $\varphi_\leq^{c,c'}(x,y)$, for $c,c'\in C$. There exists a path from the $c$ copy of a node $i$ to the $c'$ copy of a node $j$ if, and only if, the formula $\varphi_\leq^{c,c'}(x/i,y/j)$ is true.
\end{itemize}
All formulas are required to be in $\FO[<]$ and are evaluated on the input graph.

The output graph is defined as a substructure of the $C$ copies of the input linear graph, in which a node exists if it satisfies one position formula, and is labelled accordingly, and the order is defined according to the order formulas.
\end{definition}

In this article, we are only interested in linear graph transductions,
which only accept words seen as linear graphs as input.  An input word has
an \emph{image} by a $\FO$ graph transduction if the associated linear graph
satisfies its domain formula and the order relation of the output graph,
defined by the order formulas, defines a linear graph corresponding to a
word. 
If one condition fails, then the function is undefined on the given input.
 One should note that the fact that a graph is linear and corresponds
to a word is $\FO$-definable.

In Figure~\ref{Ex-onLogic}, we give the output structure of $T$ over the linear graph $u=aababb$. Note that for the sake of readability, we do not draw the whole order relation, but simply the successor relation.

\section{Main result} \label{sec:main}

We are now ready to state the main result of this article, as an extension
of the result by McNaughton and Papert \cite{MP71} and Schützenberger
\cite{Schutzenberger65} in the context of two-way transducers and $\MSO$
transductions established by Engelfriet and Hoogeboom \cite{EH01}.

\begin{theorem}\label{THM-main}
  The functions realized by aperiodic two-way transducers are exactly the
  functions realized by $\FO$ graph transductions over words.
\end{theorem}

The theorem is proved in Sections~\ref{sec:2WtoFOT}
and~\ref{sec:FOTtoA2W}.  The first inclusion relies on
Theorem~\ref{2w-FOT:Result}, while the second inclusion stems from the
conjunction of Theorems~\ref{FOT-2w:FOTto2wFO},~\ref{FOT-2w:2wFOto2wSF}
and~\ref{FOT-2w:2wSFto2w}.
The next Section is devoted to the composition of transducers, which is a key tool of the proof.

\section{Composition of transducers}\label{sec:composition}

As transducers realize functions over words, the natural question of the
compositionality occurs.  In a generic way, this question is : given two
functions realized by some machine, can we construct a machine that
realizes the composition of these functions.  This question has been
considered in~\cite{HU67} for generic machines, and resolved positively in
the case of deterministic two-way transducers in~\cite{CJ77}.

This result can also be obtained using the equivalence of two-way
transducers with $\MSO$ transductions, since these are easily proved to be
stable by composition (see \cite{Cou94}).  However, the reduction from
$\MSO$ transductions to two-way transducers established in~\cite{EH01}
makes an extensive use of a weaker version of this result, which is that
the composition of a one-way deterministic, called sequential in the following, transducer with a two-way transducer can be
done by a two-way transducer, which was first proved in~\cite{AHU69}.

In this section, we follow this approach, and now prove that this result
holds for aperiodic transducers, in the sense that if the two input
transducers are aperiodic, then we can construct an aperiodic transducer
realizing the composition.

\begin{theorem}\label{Comp-Thm}
  Let $\mathcal{A}$ be a sequential transducer that can be
  composed with a two-way transducer $\mathcal{B}$, both deterministic and aperiodic. Then we can
  effectively construct an aperiodic and deterministic two-way transducer $\mathcal{C}$ such
  that $\mathcal{C}=\mathcal{B}\circ\mathcal{A}$.
\end{theorem}

\section{From aperiodic two-way transducers to $\FO$ transductions}\label{sec:2WtoFOT}

Let us consider a deterministic and aperiodic two-way transducer.
We aim to construct a first-order graph transduction that realizes the same function.

In order to do that, we need to define a formula $\varphi_{dom}$ for the input domain, formulas $\varphi_{pos}$ for each copies of a position and each output letter of $\mathcal{A}$, and, contrary to the generic case of $\MSO$ graph transductions where only the successor is defined, we need here to define order formulas $\varphi_\leq$ that describe the order relation on the output depending on the copies of the nodes from the input.

The following result simply stems from the equivalence of aperiodic monoids and first order logic established in~\cite{Schutzenberger65,MP71}, but is an essential step to link aperiodicity to first-order, as it is used in the next theorem, which proves that the order relation between positions is first-order definable.

\begin{lemma}\label{2w-FOT:locallanguages}
Let $\mathcal{A}=(Q,A,\delta)$ be an aperiodic two-way automaton.
Then the relation classes of $\LL$, $\LR$, $\RL$, $\RR$  and consequently $\sim_\mathcal{A}$ of $\mathcal{A}$ are $\FO$-definable.
\end{lemma}

\begin{lemma}\label{2w-FOT:Order-formula}
Let $\mathcal{A}$ be an aperiodic two-way automaton.
Then for any pair of states $q$ and $q'$ of $\mathcal{A}$, there exists a $\FO$-formula 
$\varphi^{q,q'}(x,y)$ such that for any word $u$ in the domain of $\mathcal{A}$ and any pair of positions $i$ and $j$ of $u$, 
$$u\models \varphi^{q,q'}(x/i,y/j)$$
if, and only if, the run of $\mathcal{A}$ over $u$ starting at position $i$ in state $q$ eventually reaches the position $j$ in state $q'$.

\end{lemma}

We now state the main result of this section and construct the first-order transduction that realizes $\mathcal{A}$.

\begin{theorem}\label{2w-FOT:Result}
Let $\mathcal{A}$ be an aperiodic two-way transducer.
Then we can effectively construct a $\FO$-graph transduction that realizes the same function as $\mathcal{A}$.
\end{theorem}

\begin{proof}
For simplicity of the proof, we consider a transducer $\mathcal{A}=(Q,A,B,\delta,\gamma,i,F)$  where the production of any transition is at most one letter.
This can be done without loss of generality, since any given transducer can be normalized this way by increasing the number of states.
We now give the formal definition of the $\FO$ transduction ${T=(A,B,\varphi_{dom},Q,\varphi_{pos},\varphi_\leq)}$ that realizes $\mathcal{A}$.

As we consider string transductions within the scope of graph transductions, the domain formula also has to ensure that the input is a linear graph. This can be done in $\FO$ by a formula stating that there is one position that has no predecessor, one position that has no successor, every other position has exactly one successor and one predecessor and every pair of positions is comparable.
Then the domain formula of $T$ is the formula describing the language recognized by the input automaton of $\mathcal{A}$ conjuncted with the linear graph formula. By Lemma~\ref{2w-FOT:locallanguages}, as $\mathcal{A}$ is aperiodic the domain formula is $\FO$-definable. 
The order formulas are given by Lemma~\ref{2w-FOT:Order-formula}, where obviously $\varphi_\leq^{q,q'}(x,y)=\varphi^{q,q'}(x,y)$.

The $\varphi_b^q(x)$ formulas, where $q\in Q$ and $b\in B$, 
 express that the production of $\mathcal{A}$ at the position quantified by $x$ in state $q$ is $b$, 
but also that the run of $\mathcal{A}$ over $u$ reaches the said position in state $q$.
Should we define $A_{b,q}=\{a\in A\mid \gamma(a,q)=b\}$, then the first condition is expressed as $\bigvee_{a\in A_{b,q}} a(x)$.
The second condition is then equivalent to saying that there exists a run from the initial state of $\mathcal{A}$ to the current position, which is expressed by the formula $\exists y \forall z\ y\leq z\wedge \varphi^{i,q}(y,x)$.
The formula $\varphi_b^q(x)$ is thus defined as the conjunction of these two formulas.

The transduction $T$ is now defined. All formulas are expressed in the first order logic, and it realizes the same function as $\mathcal{A}$, proving the theorem.
\end{proof}

\section{From $\FO$ transductions to aperiodic two-way transducers}\label{sec:FOTtoA2W}

The proof scheme for this inclusion is adapted from the one in~\cite{EH01} proving that $\MSO$ transductions are realized by two-way deterministic transducers.
We prove that we can construct an aperiodic two-way transducer with $\FO$ look around from a $\FO$ transduction, and that the constructions given in~\cite{EH01} suppressing the look around part preserve the aperiodicity.

We define in the next subsection the models of transducers with look-around that are used in the proof.
We then give an alternative definition of aperiodicity which can be applied to transducers with logic look-around before explaining the constructions that lead up to the result.

\subsection{Transducers with look-around}

Here, we define two kinds of transducers with look-around.
The first one is a restriction of two-way transducers with regular look-around, where we limit the regular languages used in the tests to Star-free languages, which is the rational characterization of first-order logic.
These transducers differ from the classic ones by their transitions, where the tests are not determined by the letter read, but also by the prefix and suffix which can be evaluated according to some regular languages.

The second extension we consider is transducers with first-order look around.
In this case, the selection of a transition, as well as the movements of the reading head, are determined by formulas.
Formal definitions are given below.

In both cases only the definition of transition is changed, the definition of run and accepting run remaining the same.

\begin{definition}[two-way transducer with Star-Free look around]
 \emph{Two-way transducers with Star-Free look around} are a subclass of two-way transducers with regular look around defined in \cite{EH01}, where all languages in the tests are Star-Free.

Formally, it is a machine $\mathcal{A}=(Q,A,B,\Delta,i,F)$ where $Q$, $A$, $B$, $i$ and $F$ are the same as for two-way transducers,
and transitions and productions are regrouped in $\Delta$, and are of the form 
$(q,t,q',v,m)$ where $q$ and $q'$ are states from $Q$, $v\in B^*$ is the production of the transition, $m\in\{-1,0,+1\}$ describes the movement of the reading head
and $t$ is a test of the form $(L_p,a,L_s)$ where $a$ is a letter of $A\uplus\{\vdash,\dashv\}$, and $L_p$ and $L_s$ are Star-Free languages over the same alphabet.
A test $(L_p,a,L_s)$ is satisfied if the reading head is on a position labelled by the letter $a$, 
the prefix of the input word up to the position of the reading head belongs to $L_p$, and symmetrically the suffix belongs to $L_s$.

Such a machine is deterministic if the tests performed in a given state are mutually exclusives.

\end{definition}

\begin{definition}[two-way transducer with $\FO$ look around]
 \emph{Two-way transducers with $\FO$ look around} are a subclass of two-way transducers with $\MSO$ look around where formulas are restricted to the first-order.

Formally, it is a machine $\mathcal{A}=(Q,A,B,\Delta,i,F)$ where $Q$, $A$, $B$, $i$ and $F$ are the same as two-way transducers,
and transitions of $\Delta$ are of the form 
$(q,\varphi(x),q',v,\psi(x,y))$  where $q$ and $q'$ are states from $Q$, $v\in B^*$ is the production of the transition and $\varphi(x)$ and $\psi(x,y)$ are $\FO$ formulas with respectively one and two free variables.
A transition $(q,\varphi(x),q',v,\psi(x,y))$ can be taken if the formula $\varphi(x)$ holds on the input word, where $x$ quantifies the current position $i$ of the reading head, say
$\vdash\! u\!\dashv \ \models \varphi(x/i)$.
Then the reading head moves to a position $j$ such that
 $\vdash\! u\!\dashv \ \models \psi(x/i,y/j)$.	

Such a machine is deterministic if the unary tests appearing in a given state are mutually exclusive, and if for any input word $u$, any movement formula $\psi(x,y)$ and any position $i$, there exists at most one position $j$ such that 
 ${\vdash\! u\!\dashv \ \models \psi(x/i,y/j)}$.

\end{definition}

\subsection{Aperiodicity by path contexts}

The reading head of transducers with logic look-around can jump several positions at a time and in any direction. Then the notion of behavior for such transducers becomes blurry, since behaviors would have to be considered starting at any position, and moreover the direction taken while exiting a word is not decided locally, but depends on the context.

We thus give an equivalent characterization of the aperiodicity of a transducer through all contexts at a time, for machines whose reading head does not move position by position.

We recall that given a (deterministic) transducer ${\mathcal{A}=(Q,A,B,\Delta,init,F)}$ and $u$ an input word of $\mathcal{A}$, the accepting path of $\mathcal{A}$ over $u$, denoted $path(u)$, is
the sequence  $(q_0,i_0)\ldots (q_n,i_n)$ of pairs from $Q\times[0,|u|+1]$ (the length of $u$ plus the endmarkers) describing the behavior of the reading head of $\mathcal{A}$ while reading $u$, as defined in Subsection~\ref{Subsection:2w}.

We now define the projection of paths, as a way to highlight some information and forget the rest. It is applied to contexts in order to only retain the influence of a word on its context.

\begin{definition}[Projection and context paths]

Let ${I=[i_1,\ldots,i_k]}$ be an ordered sequence of integers.
We define $path_I(u)$ as the sequence of pairs from $Q\times\{1,\ldots,k\}$ such that for any pairs $(q,j)$ and $(q',j')$, 
$(q,j)$ appears before $(q',j')$ in $path_I(u)$ if, and only if, 
$(q,i_j)$ appears before $(q',i_{j'})$ in $path(u)$.
Informally, this corresponds to selecting pairs whose position is in $I$ 
and renaming them according to the set $I$.

Abusing notations, we will note the \emph{context path} $path_{vw}(vuw)=path_I(vuw)$ where $I$ is the set of positions of $v$ and $w$.
\end{definition}

Then $path_{vw}(vuw)$ is the trace of the run over $u$ on the context $v,w$ and two words $u$ and $u'$ are $\mathcal{A}$-equivalent if for any context $v,w$, we have equality of the paths contexts $path_{vw}(vuw)=path_{vw}(vu'w)$.
Then by definition of the aperiodicity, 
a transducer or an automaton $\mathcal{A}$ is \emph{aperiodic} if
there exists a positive integer $n$ such that for any words $u$, $v$ and $w$ on the input alphabet of $\mathcal{A}$, the context paths
 $path_{vw}(vu^nw)$ and $path_{vw}(vu^{n+1}w)$ are equals.
One should remark that on two-way transducers, this notion is equivalent to the aperiodicity of the transition monoid.
The next lemma serves as the link from the first-order logic to the aperiodicity by context paths.

\begin{lemma}\label{FOT-2w:ParcFOT}
Let $T$ be a $\FO$ graph transduction.
There exists a positive integer $n$ such that for any input words $u$, $v$ and $w$ such that $vu^nw$ is in the domain of $T$,
$vu^{n+1}w$ is also in the domain of $T$ and the two words satisfy the same formulas of $T$, when the free variables quantify positions of $v$ or $w$.

\end{lemma}

\begin{proof}
First consider the domain formula of $T$. Since it is a $\FO$ formula, it has an aperiodicity index $n$, in the sense that for any  words $u$, $v$ and $w$, $vu^nw$ is in the domain of $T$ if, and only if, $vu^{n+1}w$ is in the domain of $T$.

We now prove the result in the case where $i$ ranges over $v$ and $j$ ranges over $w$, but similar proofs hold for $i$ and $j$ ranging independently over $v$ and $w$.
Consider a pair $c$, $c'$ of copies in $C$, and integers $0\leq i<|v|$ and $0\leq j < |w|$.
Then a word with positions $i$ and $j$ quantified respectively by $x$ and $y$ can be seen
 as a word over the alphabet $A\times\{0,1\}^2$, where all letters have $(0,0)$ as second component, except $(v_i,1,0)$ and $(w_j,0,1)$.
The formula $\varphi_\leq^{c,c'}(x,y)$ can then be equivalently seen as a closed formula over this enriched alphabet.
 This formula being in $\FO$, it describes an aperiodic language, and then there exists an integer $n'$ such that
 $vu^{n'}w$
  satisfies  $\varphi_\leq^{c,c'}(x/i,y/|vu^{n'}|+j)$ if, and only if,
 $vu^{n'+1}w$ satisfies $\varphi_\leq^{c,c'}(x/i,y/|vu^{n'+1}|+j)$.
 
A similar argument also holds for the node formulas $\varphi_b^c(x)$.
 As there is a finite number of formulas, there exists an integer, the maximum of the index of each formulas, such that the result holds. 
\end{proof}

This lemma means that the transducer has an aperiodicity index, in the sense that $u^n$ and $u^{n+1}$ behave the same way for the same context.
It also corresponds to the notion of aperiodicity defined earlier in this section, where the sequence ranges over pairs of copy and position.

\subsection{Construction of the aperiodic transducer}
We now hold all the necessary tools to prove the reduction from $\FO$ transductions to aperiodic two-way transducers.

We present the first construction, from a $\FO$ transduction to a two-way transducer with $\FO$ look around.
The construction is quite simple. 
By putting the copy set of the graph transduction as the set of states of the transducer, we can use the fact that the reading head of a transducer with logic look around jumps between positions to strictly follow the output structure of the input transduction. 
We then use Lemma~\ref{FOT-2w:ParcFOT} to prove the aperiodicity of the construction.

\begin{theorem}\label{FOT-2w:FOTto2wFO}
Let $T$ be a $\FO$ graph transduction.
Then we can effectively construct an aperiodic two-way transducer with $\FO$ look around that realizes the same function over words.
\end{theorem}

We have now constructed an aperiodic two-way machine from the input $\FO$ transduction.
But even though two-way transducers with $\MSO$ look around are known to be equivalent to two-way transducers \cite{EH01}, we need to prove that we can suppress the $\FO$ look around while preserving the aperiodicity of the construction.
This is done by the two following theorems, using Star-free look around as an intermediate step.
We show that the construction evoked in~\cite{EH01} do preserve the aperiodicity, leading to the result.

\begin{theorem}\label{FOT-2w:2wFOto2wSF}
Given an aperiodic two-way transducers with $\FO$ look around, we can construct an
aperiodic two-way transducers with Star-Free look around that realizes the same function.

\end{theorem}

\begin{proof}
In order to prove this theorem, we rely on the proof of Lemma~6 from~\cite{EH01}, which proves that two-way transducers with $\MSO$ look around can be expressed by two-way transducers with regular look around.
This is done by constructing a transducer whose regular tests stem directly from the $\MSO$ formulas.
Then the reading head simulates the jumps of the reading head of the transducer with $\MSO$ look around by moving step by step up to the required position.

We aim to prove on one hand that if the formulas are defined in the first-order, then the resulting two-way transducer only uses Star-free look around, 
and on the other hand that if moreover the input transducer is aperiodic, then the output transducer is also aperiodic.

The first claim is proved by noticing that the languages used in the regular look around construction are languages defined by formulas of the input transducer with $\FO$ look around with free variables. Then if the free variables are seen as an enrichment of the alphabet, similarly to what is done in the proof of Lemma~\ref{FOT-2w:ParcFOT}, the formula remains first-order, and consequently all the languages used in look around tests are Star-free.

Now let us compare the moves of the reading head of the resulting transducer with Star-free look around with the ones of the head of the input transducer with $\FO$ look around.
The path of the resulting transducer over any word can entirely be deduced from the path of the input transducer, by adding step by step walks between the jumps of the reading head.
Then, if the input transducer is aperiodic with index $n$, given three words $u$, $v$ and $w$, the context paths $path_{vw}(vu^nw)$ and $path_{vw}(vu^{n+1}w)$ are equal, and thus the context paths for the transducer with Star-free look around, that are deduced from it, are equal too, proving the aperiodicity of the resulting transducer.
\end{proof}

We finally prove that we can suppress the Star-free look around tests while preserving the aperiodicity, which concludes the proof of the main theorem.

\begin{theorem}\label{FOT-2w:2wSFto2w}
Given an aperiodic two-way with Star-free look around, we can construct an aperiodic two-way transducers that realizes the same function.
\end{theorem}

\begin{proof}
Here again, we consider the construction from~\cite{EH01}, Lemma~4, proving that we can suppress the regular look around tests.
Our goal is then to prove that this operation preserves the aperiodicity when the input transducer only uses Star-free languages.

The construction relies heavily on the fact that the composition of a two-way transducer with a one-way transducer can be done by a two-way transducer.
It is used to preprocess the input by adding the result of each regular tests from the transitions at each position.
Given the test $(L_p,a,L_s)$ of a transition, a left-to-right pass simulates $L_p$ and reproduces the input where each position is enriched with the information~: does its prefix belong to $L_p$.
Symmetrically, a right-to-left transducer adds the same information for $L_s$.
Let us remark that since these languages are Star-free, the input automaton of the transducers simulating these languages are aperiodic.

Then the information regarding every transition is added to the input, and lastly we can construct a two-way transducer that acts in the same way as the input transducer, but where all the look around have been suppressed and are done locally by looking at the enrichment part of the letter.
This two-way transducer without look around is aperiodic if the input transducer is aperiodic, since they share the same paths, and thus context paths.

Finally, the input transducer is given as the composition of a single aperiodic two-way transducer with a finite number of aperiodic one-way transducers. 
Should we first remark that, by symmetry of the problem, Theorem~\ref{Comp-Thm} also holds for right sequential transducers, through several uses of this composition result we finally obtain a unique two-way transducer that realizes the input transducer 
with Star-free look around.
\end{proof}

\section{Conclusion}

We recall that a similar work has been done for streaming string transducers by Filiot, Krishna and Trivedi~\cite{FKT14}.
Then through $\FO$ transductions, this result and Theorem~\ref{THM-main} prove the equivalence of aperiodicity for the two models of transducers.

There exists algorithms that input a two-way transducer and construct directly an equivalent streaming string transducer (see~\cite{AFT12} for example).
It would be interesting to check first if the aperiodicity is preserved through these algorithms, and secondly to compare the size and aperiodicity indexes of the two transition monoids.
Although unknown from the authors, a reciprocal procedure and its study would hold the same interest.

On a more generic note, one can ask which fragments of logic preserve their algebraic characterization in the scope of two-way transducers and $\MSO$ transductions.
For example, are $\mathcal{J}$-trivial transducers equivalent to $\mathcal{B}\Sigma_1$ transductions ?
The main challenges for this question are the stability by composition of these restricted classes of transducers on one hand, and on the other hand
the very definition of logic transductions for restricted fragments, as a fragment must retain some fundamental expressive properties, such as being able to characterize linear graphs.

Finally, we would like to point out the fact that even if we can decide if a given two-way transducer is aperiodic, it is still open to decide if the function realized by a two-way transducer can be realized by an aperiodic one.
An promising approach for this problem might be to consider machine-independent descriptions of functions, as defined recently for streaming string transducers in~\cite{AFR14} for example. This was successfully done in~\cite{Boj14} for machines with origin semantic.
We also think that this question could be solved by the notion of canonical object of a function over words, which has yet to be defined.

\paragraph*{Acknowledgements}

We would like to thank Antoine Durand-Gasselin, Pierre-Alain Reynier
and Jean-Marc Talbot for very fruitful discussions.

\bibliography{IEEEabrv,ap2way}

\newpage
\section*{Appendix}
\appendix

\subsection*{Composition of transducers}
		\setcounter{theorem}{9}
\begin{theorem}
Let $\cA$ be a one-way transducer and $\cB$ be a two-way transducer, both deterministic and aperiodic, that are composable.
Then we can effectively build a deterministic and aperiodic two-way transducer $\cC$ that realizes the function $\cB\circ\cA$.
\end{theorem}

\begin{proof}
We first describe precisely the construction evoked in \cite{CJ77}.
The point here is to give solid bases to then prove the aperiodicity of this construction.

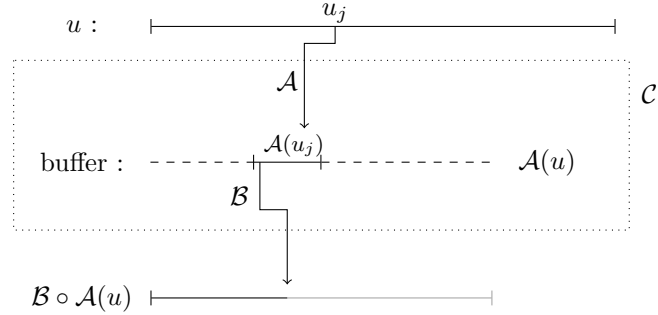
\begin{figure}[h]
\begin{center}
\begin{tikzpicture}[scale=0.9]

\node (u) at (0,5) {$u$ :}; 
\draw[|-|] (1,5)--(7.8,5);
\node (au) at (6.8,3) {$\mathcal{A}(u)$};
\node (buff) at (0,3) {buffer :};

\draw[dotted] (-1,4.5) rectangle (8, 2);
\node (bchap) at (8.3,4) {$\mathcal{C}$};
\node (aau) at (0,1) {$\mathcal{B}\circ\mathcal{A}(u)$};

\node(u4) at (3.7,5.2) {$u_j$};

\draw[dashed] (1,3) -- (2.5,3);
\draw[|-|] (2.5,3) -- (3.5,3);
\draw[dashed] (3.5,3) -- (6,3);

\node (gqa) at (3.1,3.25) {\footnotesize{$\mathcal{A}(u_j)$}};

\draw[->] plot coordinates { (3.7,5) (3.7,4.75) (3.25,4.75) (3.25,3.5)};

\node (ca) at (3,4.2) {$\mathcal{A}$};

\draw[-|,color=gray!70] (3,1) -- (6,1);
\draw[|-] (1,1) -- (3,1);

\node (cb) at (2.3,2.5) {$\mathcal{B}$}; 
\draw[->] plot coordinates {(2.6,3) (2.6,2.3) (3,2.3) (3,1.2)};

\end{tikzpicture}
\end{center}
\caption{Schematic representation of the composition transducer $\mathcal{C}$.
The buffer always contains a production of $\mathcal{A}$ which serves as input for simulating $\mathcal{B}$.}\label{Comp-SchemaTransdC}
\end{figure}

Thus let us describe the transducer $\cC$.
Its set of states will be union of sets describing working modes of the transducer. One mode will be the easy part, when the buffer is full or when we can fill the buffer easily. The second mode will occur if we need to do a backward step on the sequential transducer. It stores the possible candidates then moves back until it can determine the right state. It then switches to the third mode to move back to the required position. The second and third working modes are illustrated in Figures~\ref{Comp-Rewinding} and~\ref{Comp3mod} respectively.\\

Formally, let $\cA=(Q,A,B,\delta,\gamma,q_0,F)$ and $\cB=(P,B,C,\alpha,\beta,p_0,G)$ be a required. We assume that both of the transducers are normalized.
Let now $\cC=(R,A,C,\mu,\eta,r_0,R_F)$ be defined as follows:

\begin{itemize}

\item $R=R_1\uplus R_2\uplus R_3$ where
		\begin{itemize}
		\item $R_1=P\times Q\times B^{\leq m+1}\times B^{\leq m+1}$ where $m$ is the maximal length of a production of $\gamma$.
		\item $R_2=P\times Q\times B\uplus P\times B\times  2^{Q\times Q}$. This set is two-part. The first part will be used at the first step of the backward computation, to get the candidates. The second part consists of the information we need to store, plus a relation over the states of $Q$. This relation describes the set of states that has a run from the current position to one of the candidates. By determinism of $\cA$, one state can not have a path to two different candidates.
		\item $R_3=P\times B\times Q\times Q$. This set consists of the same information we need to store, plus two concurrent states that will allow us to track the position we need to get back to.
		\end{itemize}
		
\item $r_0=(p_0,q_0,\epsilon,\vdash)$.
\item $R_F=\{(p,q,u,v)\mid p\in F,q\in G\}$.
\end{itemize}

We now describe $\mu$ and $\eta$ in a succession of cases:
\begin{enumerate}
\item Let $r=(p,q,u,bv)$ be a state in $R_1$, $a$ be a letter of $A$, and let $(p',d)=\alpha(p,b)$. We now treat subcases depending on $d$, $u$ and $v$: 
		\begin{enumerate}
		\item if $d=0$, then $\mu(r,a)=((p',q,u,bv),0)$,
		\item if $d=+1$ and $v\neq \epsilon$, then $\mu(r,a)=((p',q,ub,v),0)$,
		\item if $d=-1$ and $u=u'b'\neq \epsilon$, then $\mu(r,a)=((p',q,u',b'bv),0)$.\\
		In these three subcases, $\eta(r,a)=\beta(p,b)$. It is the only cases where $\eta$ is non-empty. Thus in the following we will omit $\eta$.
		\item if $d=+1$ and $v=\epsilon$, let $q'=\delta(q,a)$ and $v'=\gamma(q,a)$. Then $\mu(r,a)=((q',p,\epsilon,bv'),+1)$.
		When the buffer is empty but $\cB$ is moving forward, then we just compute $\cA$ to fill the buffer.
		\item if $d=-1$ and $u=\epsilon$, then $\mu(r,a)=((p,q,b),-1)$ with $(p,q,b)\in R_2$.
		If we need to go backward, we enter the first part of mode 2 to compute the candidates for the previous state.
		\end{enumerate}

\item Let $r=(p,q,b)$ be a state in $R_2$ and $a$ a letter of $A$.
		Now let $Q'\subseteq Q$ be the set of states $q'$ such that $\delta(q',a)=q$.
		\begin{itemize}
		\item if $Q'=\{q'\}$, then let $v=\gamma(q',a)$ and $\mu(r,a)=((p,q',v,b),0)$,
		\item otherwise, let $Rel=\{(q',q')\mid q'\in Q\}$ and
		$\mu(r,a)=((p,b,Rel),-1)$.
		\end{itemize}
		We compute here the possible candidates. If there is only one, then we can decide, otherwise we register them all and start moving backward to decide.

\item Let $r=(p,b,Rel)\in R_2$, and $a$ be a letter of $A$.
	Let now $Nrel=\{(q_1,q_2)\mid (q_1,\delta(q_2,a))\in Rel\}$.
		\begin{itemize}
		\item if there exists only one state $q$ such that $\{q\}\times Q\cap Nrel\neq \emptyset$, then let $q'\neq q$ be such that we can find $q_1$ and $q_2$ with $(q,q_1)$ and $(q',q_2)$ in $Rel$. Then $\mu(r,a)=((p,b,q_1,q_2),+1)$ with $(p,b,q_1,q_2)\in R_3$.\\
		We compute the new set of states that have a path to our candidates with one more step backward. If there now exists only one candidate left, we follow this path in mode 3, together with an other option, to be able to get back to the required position.
		\item otherwise, if $a=\vdash$, we proceed as in the previous case with $q$ such that $(q,q_0)\in Nrel$.\\
		If we reached the beginning of the input, then the right candidate is the one with a path starting at $q_0$.
		\item otherwise, then $\mu(r,a)=((p,b,Nrel),-1)$.\\
		If we can not decide now, we continue our way back.
		\end{itemize}		 

\item Let $r=(p,b,q_1,q_2)$ be a state in $R_3$ and $a$ be a letter of $A$.
	\begin{itemize}
	\item if $\delta(q_1,a)\neq \delta(q_2,a)$, then $\mu(r,a)=(p,b,\delta(q_1,a),\delta(q_2,a))$.\\
	If the two paths do not collide next step, then we are not at the required position yet, so we continue computing the paths.
	\item if $\delta(q_1,a)=\delta(q_2,a)$ then let $v=\gamma(q_1,a)$ and
		$\mu(r,a)=((q_1,p,v,b),0)$.\\
		If they do collide, then $q_1$ is the right candidate and we can fill the buffer.
	\end{itemize}
	
\item If $r=(g,q,u,v)$ with $g\in G$ then let $(q',d)=\delta(q,a)$ and
		$\mu(r,a)=((g,q',\epsilon,\epsilon),d)$.\\
		If we finished the computation on $\cB$, we simply finish the one on $\cA$ to ensure that the function realized by $\cA$ is defined on the input.
\end{enumerate}

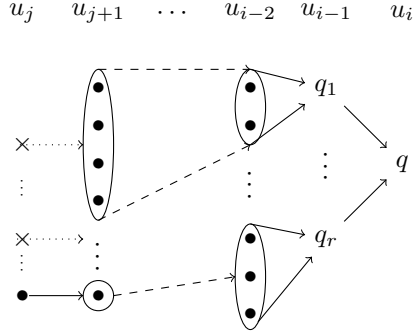
\begin{figure}
  \begin{center}
  \begin{tikzpicture}[initial text=]

\node (ui) at (0,3) {$u_j$};
\node (ui1) at (1,3) {$u_{j+1}$};
\node () at (2,3) {$\cdots$};
\node (uj2) at (3,3) {$u_{i-2}$};
\node (uj1) at (4,3) {$u_{i-1}$};
\node (uj) at (5,3) {$u_i$};

\node (q) at (5,1) {$q$};
\node (q1) at (4,2) {$q_1$};
\node (q2) at (4,0) {$q_r$};
\node[rotate=90] () at (4,1) {$\cdots$};

\draw[->] (q1) -- (q);
\draw[->] (q2) -- (q);

\node (q11) at (3,2) {$\bullet$};
\node (q12) at (3,1.5) {$\bullet$};
\node[rotate=90] () at (3,0.75) {$\cdots$};

\node (q21) at (3,-1) {$\bullet$};
\node (q22) at (3,-0.5) {$\bullet$};
\node (q23) at (3,0) {$\bullet$};

\draw (3,1.75)ellipse(0.2 and 0.5);
\draw[->] (q11)+(0,0.25) -- (q1);
\draw[->] (q12)+(0,-0.25) -- (q1);

\draw (q22)ellipse(0.2 and 0.7);
\draw[->] (q23)+(0,0.2) -- (q2);
\draw[->] (q21)+(0.1,-0.1) -- (q2);

\node (p11) at (1,2) {$\bullet$};
\node (p12) at (1,1.5) {$\bullet$};
\node (p13) at (1,1) {$\bullet$};
\node (p14) at (1,0.5) {$\bullet$};
\node[rotate=90] () at (1,-0.2) {$\cdots$};

\node (p21) at (1,-0.75) {$\bullet$};

\draw (p13)+(0,0.25)ellipse(0.2 and 1);
\draw (p21)circle(0.2);

\draw[->,dashed] (p11)+(0,0.25) -- (3,2.25);
\draw[->,dashed] (p14)+(0,-0.25) -- (3,1.25);
\draw[->,dashed] (p21)+(0.2,0)-- (2.8,-0.5);

\node (no2) at (0,0) {$\times$};
\draw[dotted,->] (no2)+(-0.07,0) -- (0.8,0);
\node (no) at (0,1.25) {$\times$};
\draw[dotted,->] (no)+(-0.033,0) -- (0.8,1.25);
\node[rotate=90,scale=0.6] () at (0, 0.7) {$\cdots$};
\node[rotate=90,scale=0.6] () at (0,-0.3) {$\cdots$};

\node[minimum size=0cm,inner sep=0pt,outer sep=0pt]  (yay) at (0,-0.75) {$\bullet$};
\draw[->] (yay) -- (0.8,-0.75);

\end{tikzpicture}
  \end{center}
  \caption{The transducer $\mathcal{C}$ reads from position $i-1$ to $j$
  			calculating the runs of $\mathcal{A}$ leading to each potential state.
  			Here, the correct candidate is state $q_r$, decided at position $j$.}\label{Comp-Rewinding}
\end{figure}

\begin{figure}
  \Large
  \begin{center}
  \begin{tikzpicture}[initial text=]

\node (ui) at (0,3) {$a_i$};
\node (ui1) at (1,3) {$a_{i+1}$};
\node (dot) at (2,3) {$\cdots$};
\node (uj1) at (3,3) {$a_{j-1}$};
\node (uj) at (4,3) {$a_j$};

\node (q2) at (1,2) {$q_2^k$};
\node (dot2) at (2,2) {$\cdots$};
\node (qp) at (3,2) {$q_2$};

\node (trouve) at (0,1) {};
\node (q1) at (1,1) {$q_1^k$};
\node (dot1) at (2,1) {$\cdots$};
\node (q) at (3,1) {$q_1$};
\node (origine) at (4,1) {$q$};

\draw[->] (q2) -- (dot2);
\draw[->] (dot2) -- (qp);
\draw[->] (qp) -- (origine);
\draw[->] (trouve) -- (q1);
\draw[->] (q1) -- (dot1);
\draw[->] (dot1) -- (q);
\draw[->] (q) -- (origine);
\end{tikzpicture}
  \end{center}
  \caption{Backwards moves using states from $R_3$.}\label{Comp3mod}
\end{figure}
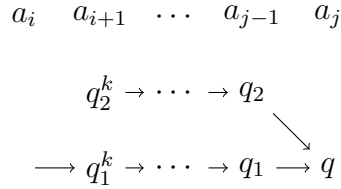

Let us now prove that aperiodicity is preserved by this construction, meaning that if both $\cA$ and $\cB$ are aperiodic transducers, then the resulting $\cC$ will also be aperiodic. We denote by $n_\cA$ and $n_\cB$ the aperiodicity indexes of $\cA$ and $\cB$ respectively, and let us prove that $\cC$ is aperiodic with index $n=2n_\cA+n_\cB+1$. \\
We first make a few remarks:
\begin{enumerate}
\item\label{Cap-R1avant}The transducer $\cC$ can only go forward in states of $R_1$ or $R_3$. Furthermore, if a run is in a state of $R_1$, then any further position visited by the run will be first reached in a state of $R_1$. This is due to the construction of $\cC$, the states of $R_3$ are used to go back to the last position before a position reached by a state of $R_1$.
\item\label{Cap-R2arr}The transducer $\cC$ can only move backward in states of $R_1$ or $R_2$. But since moving backward in a state of $R_1$ leads to a state of $R_2$, the transducer only reaches previous positions in states of $R_2$.
\end{enumerate}
We now treat the case of each congruence separately.
\begin{itemize}
\item $u^n\LR u^{n+1}$.
	Let assume there exists a run starting at a state $r$ on the left of $u^n$ and going out of it on the right in state $r'$. Note that the following proof also holds if we consider a run over $u^{n+1}$.
\begin{itemize}
\item If $r$ is in $R_1$, then, by \ref{Cap-R1avant}, $r'$ is also in $R_1$. Moreover, each iteration of $u$ is first reached by a state in $R_1$.
	Thus there exists an underlying run of $\cA$ on $u^n$, and after $n_\cA$ iterations of $u$, the run first enters $u$ in a state with the same $Q$ component $q$.
	Then, we look at the production $v$ of $\cA$ over a left-to-right run over $u$ starting and ending in state $q$.
	If $v=\epsilon$, then the run is just advancing waiting to fill its buffer. Thus the same run exists for $u^{n+1}$.
	Otherwise, then after $n_\cB$ more iterations, $\cB$ would have worked on $v^{n_\cB}$ and hence the $P$ component of the state is stabilized on $p$. Plus, if $v=v'b$, then the buffers will be $(v',b)$. Then at each iteration the run first exits $u$ in the same state $(p,q,v',b)$. Thus in this case we have an aperiodicity index of $n_\cA+n_\cB$.

\item If $r$ and $r'$ are in $R_3$, then the run never leaves states of $R_3$. Thus it is just two parallel runs of $\cA$. Then after $n_\cA$ iterations we get aperiodicity of the run.

\item If $r$ is in $R_3$ but $r'$ is in $R_1$, then after at most $n_\cA$ iterations we reached a state of $R_1$, otherwise we are in the previous case. We then are, after at most $n_\cA$ iterations of $u$, in the case where both $r$ and $r'$ are in $R_1$.
Thus by combining the two previous points, we get aperiodicity with an index of $2n_\cA+n_\cB$.
\end{itemize}	

\item $u^n\RL u^{n+1}$
	Let assume there exists a run starting at a state $r$ on the right of $u^n$ and going out of it on the left in state $r'$. Note that the following proof also holds if we consider a run over $u^{n+1}$.
	Then we can assume $r'$ to be in $R_2$. So the run is mainly a backward computation of $\cA$, with potentially computations over the buffer in between.
	\begin{itemize}
	\item If the run remains in states of $R_2$ long enough, say during $n_\cA$ iterations of $u$, then the states visited are of the form $(p,b,Rel)$, with $p$ and $b$ constants, while $Rel$ computes the sets of states having a run back to a given position. Thus in this case we know that for $n\leq n_\cA$ the $Rel$ component will always be the same of a given position of iterations of $u$. Then the same runs will exist for $u^n$ and $u^{n+1}$.
	\item Otherwise, it means that the run is able to compute the true candidate, and go back to the required position. Plus we know that after $n_\cA$ iterations (where we start counting from the left), the production will always be the same word $v$.
	 If $v$ is empty then the computation on $\cB$ does not move, thus the $P$ component stays the same, and we get aperiodicity of the run.
	 If $v$ is not empty, then after $n_\cB$ iterations (where we start counting from the right), $\cB$ worked on a word $v^{n_\cB}$. Then we know that far enough from the borders of the word, the run will first enters each iterations of $u$ in a same state, giving aperiodicity with index $n_\cA+n_\cB$.
	\end{itemize}
	
\item Finally, we can prove that if the relations $\LR$ and $\RL$ are aperiodic with an aperiodicity index of $n$, then the relations $\LL$ and $\RR$ are aperiodic with an aperiodicity index of $n+2$.
First remark that left-to-left and right-to-right runs over $u^{n+2}$ exists over $u^{n+3}$.
Now consider a run over an input word $u^{n+3}$.
We decompose it as follows.
We isolate one iteration of $u$ on each side of the input, and decompose the run in a succession of small runs over these isolated iterations, and runs over $u^{n+1}$.
Then we know that the transversal runs also exists over $u^n$, and the left-to-left and right-to-right runs will appear over $u^n$ plus the isolated iteration as described in Figure~\ref{Comp-Ap:LRtoLL}.

Then the transducer $\mathcal{C}$ is aperiodic, with an aperiodicity index bounded by $n_\mathcal{A}+n_\mathcal{B}+2$.
\end{itemize}
\begin{figure}

  \begin{tikzpicture}[initial text=,scale=0.65]

\node[minimum size=0cm] (debut) at (0,1) {};
\node[minimum size=0cm,inner sep=0pt,outer sep=0pt] (it1) at (2,1) {$|$};
\node[minimum size=0cm,inner sep=0pt,outer sep=0pt] (it2) at (5,1) {$|$};
\node[minimum size=0cm,inner sep=0pt,outer sep=0pt] (itna) at (7,1) {$|$};

\node[minimum size=0cm,inner sep=0pt,outer sep=0pt] (itfin) at (10,1) {$|$};
\node[minimum size=0cm,inner sep=0pt,outer sep=0pt] (itfin2) at (12,1) {$|$};
\node[minimum size=0cm,inner sep=0pt,outer sep=0pt] (itfin3) at (14,1) {};

\node (u1) at (1,1.3) {$u$};
\node (u2) at (6,1.3) {$u$};
\node (dot) at (3.5,1.3) {$\cdots$};
\node (una) at (11,1.3) {$u$};
\node (dot) at (8.5,1.3) {$\cdots$};
\node () at (13,1.3) {$u$};

\draw[|-|,thick] (debut)--(itfin3);

\node () at (-0.5,-2.9) {$\mathcal{B}$};

\draw[->] (0,0) -- (14,0);
\node () at (-0.5,0) {$\mathcal{A}$};

\footnotesize{
\node[minimum size=0cm] () at (0,0) {};
\node[minimum size=0cm,inner sep=0pt,outer sep=0pt] () at (2,0) {$|$};
\node[minimum size=0cm,inner sep=0pt,outer sep=0pt] () at (5,0) {$|$};
\node[minimum size=0cm,inner sep=0pt,outer sep=0pt] () at (7,0) {$|$};
\node[minimum size=0cm,inner sep=0pt,outer sep=0pt] () at (10,0) {$|$};
\node[minimum size=0cm,inner sep=0pt,outer sep=0pt] () at (12,0) {$|$};

}

\node () at (0.2,-0.4) {$q_1$};
\node () at (2,-0.4) {$q_2$};
\node () at (5,-0.4) {$q$};
\node () at (7,-0.4) {$q$};
\node () at (10,-0.4) {$q$};
\node () at (12,-0.4) {$q$};
\node () at (14,-0.4) {$q$};

\draw (0.1,1.6) .. controls (0.2,2.5) and (2.3,1.9) .. (2.5,2.4);
\draw (5,1.6) .. controls (4.8,2.5) and (2.7,1.9) .. (2.5,2.4);
\node (per) at (2.5,2.4) [above] {$n_\mathcal{A}$ iterations }; 

\draw (5,-0.6) -- (5,-0.8) -- (7,-0.8) -- (7,-0.6);
\draw[->] (6,-0.8) -- (6,-1.2);
\node () at (6,-1.5) {$v$};

\draw (10,-0.6) -- (10,-0.8) -- (12,-0.8) -- (12,-0.6);
\draw[->] (11,-0.8) -- (11,-1.2);
\node () at (11,-1.5) {$v$};

\node () at (8.5,-1.5) {$\ldots$};

\draw (12,-0.6) -- (12,-0.8) -- (14,-0.8) -- (14,-0.6);
\draw[->] (13,-0.8) -- (13,-1.2);
\node () at (13,-1.5) {$v$};

\draw[thick,|-|] (14,-1.9) -- (0,-1.9);
\node () at (12,-1.9) {$|$};
\node () at (10,-1.9) {$|$};
\node () at (7,-1.9) {$|$};
\node () at (5,-1.9) {$|$};

\node () at (-1.3,-1.9) {$\mathcal{A}(q_1,u^{n+1})$};
\node () at (-0.4,1) {$u^{n+1}$};

\node[minimum size=0cm,inner sep=0pt,outer sep=0pt] (p1) at (5,-2.7) {$|$};
\node () [below=0 of p1] {$p_1$};

\node[minimum size=0cm,inner sep=0pt,outer sep=0pt] (p2) at (7,-2.72) {$|$};
\node () [below=0 of p2] {$p_2$};

\node[minimum size=0cm,inner sep=0pt,outer sep=0pt] (pnm1) at (10,-2.75) {$|$};

\node[minimum size=0cm,inner sep=0pt,outer sep=0pt] (pn) at (12,-3.15) {$|$};
\node () [above=0 of pn] {$p$};

\node[minimum size=0cm,inner sep=0pt,outer sep=0pt] (pn2) at (14,-3.55) {$|$};
\node () [above=0 of pn2] {$p$};

\draw plot [smooth] coordinates {(0,-2.4) (4,-2.5) (1.3,-2.6) (p1) (p2) (pnm1) (11.2,-2.8) (9,-3) (pn) (13,-3.2) (11,-3.4) (pn2) };
\draw[->] (14,-3.55) -- ++(0.2,0.01);

\draw (5,-3.6) .. controls (5.2,-4.5) and (8.3,-3.5) .. (8.5,-4.2);
\draw (12,-3.6) .. controls (11.8,-4.5) and (8.7,-3.5) .. (8.5,-4.2);
\node () at (8.5,-4.4) [below] {$n_\mathcal{B}$ iterations}; 

\end{tikzpicture}
\caption{A left-to-right behavior of $u^n$ in $\mathcal{C}$ simulates a left-to-right behavior of $u^n$ in $\mathcal{A}$ 
and a left-to-right behavior of $\mathcal{A}(q_1,u^n)$ in $\mathcal{B}$.}\label{Comp-llrun}
\end{figure}
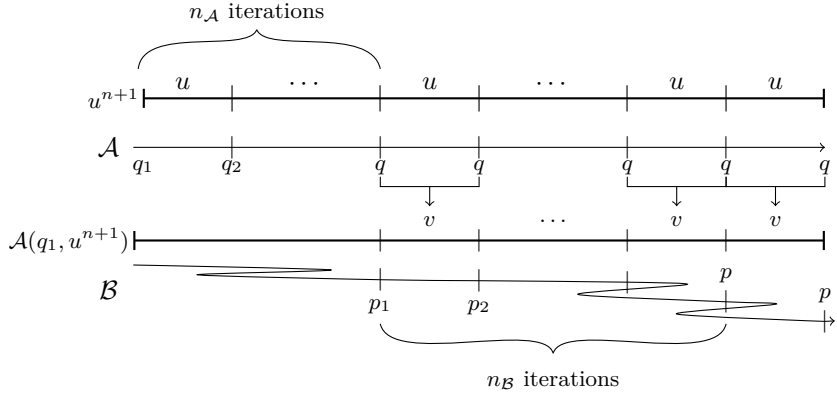

\begin{figure}
\begin{center}

\begin{tikzpicture}[scale=0.6]

\draw[|-|] (0,6) -- (1,6);
\draw[-|] (1,6) -- (4,6);
\draw[-|] (4,6) -- (5,6);
\node (u1) at (0.5,6.3) {$u$};
\node (u2) at (2.5,6.3) {$u^{n+1}$};
\node (u3) at (4.5,6.3) {$u$};

\draw[dashed] (1,6) -- (1,3.7);
\draw[dashed] (4,6) -- (4,3.7);

\draw (-0.3,5.5) .. controls (5.5,5.4) and (5.5,5.3) .. (2.5,5);
\draw (2.5,5) .. controls (1,4.8) and (1,4.7) .. (2.5, 4.5);
\draw[->] (2.5,4.5) .. controls (5.8, 4.3) and (5.9,4.2) .. (-0.3,4);

\node (to) at (6.2,5) {$\Longrightarrow$};

\draw[|-|] (8,6) -- (9,6);
\draw[-|] (9,6) -- (12,6);
\draw[-|] (12,6) -- (13,6);
\node (v1) at (8.5,6.3) {$u$};
\node (v2) at (10.5,6.3) {$u^n$};
\node (v3) at (12.5,6.3) {$u$};

\draw (7.7,5.5) .. controls (13.5,5.4) and (13.5,5.3) .. (10.5,5);
\draw (10.5,5) .. controls (8,4.8) and (8,4.7) .. (10.5, 4.5);
\draw[->] (10.5,4.5) .. controls (13.8, 4.3) and (13.9,4.2) .. (7.7,4);

\draw[dashed] (9,6) -- (9,3.7);
\draw[dashed] (12,6) -- (12,3.7);

\end{tikzpicture}
\caption{Transversal runs are shortened by aperiodicity of $\LR$ and $\RL$ and the right-to-right run overflows to the left isolated iteration.}\label{Comp-Ap:LRtoLL}

\end{center}
\end{figure}
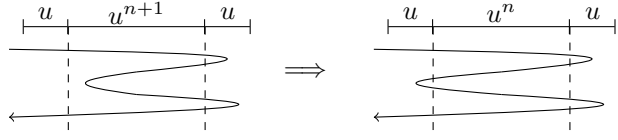
\end{proof}

\subsection*{From aperiodic two-way transducers to $\FO$ transductions}
\setcounter{theorem}{11}
\begin{lemma}
Let $\mathcal{A}$ be an aperiodic two-way automaton.
Then for any pair of states $q$ and $q'$ of $\mathcal{A}$, there exists a $\FO$-formula 
$\varphi^{q,q'}(x,y)$ such that for any word $u$ in the domain of $\mathcal{A}$ and any pair of positions $i$ and $j$ of $u$, 
$$u\models \varphi^{q,q'}(x/i,y/j)$$
if, and only if, the run of $\mathcal{A}$ over $u$ starting at position $i$ in state $q$ eventually reaches the position $j$ in state $q'$.

\end{lemma}
\begin{proof}
Without loss of generality we assume that $i$ is smaller or equal to $j$.
The theorem is proved by decomposing runs of $\mathcal{A}$ over $u$ in partial runs over $u[1,i-1]$, $u[i,j]$ and $u[j+1,|u|]$ as shown in Figure~\ref{2w-FOT:FigRun}.
Now let us remark that if the reading head is at position $i$ in a state $p$, then the next state in which the run reaches the positions $i$ or $j$ can be decided thanks to $\bh_{\ell r}(u[i,j])$, $\bh_{\ell\ell}(u[i,j])$ and
$\bh_{rr}(u[1,i-1])$.
Similarly if the reading head is at position $j$, then the next state depends on $\bh_{r\ell}(u[i,j])$, $\bh_{rr}(u[i,j])$ and 
$\bh_{\ell\ell}(u[j+1,|u|])$.
Then the equivalence classes of these three factors of $u$ decide in which states the position $j$ is visited by the run starting at position $i$ in a given state.
Thus the specification of $\varphi^{q,q'}(x,y)$, which is the fact that the run of $\mathcal{A}$ over $u$ starting at position $i$ in state $q$ eventually visits the position $j$ in state $q'$, entirely relies on the relation classes of these three factors of $u$.

Thanks to Lemma~\ref{2w-FOT:locallanguages}, we know that given $\mathcal{A}$ aperiodic, there exists $\FO$ formulas describing each behavior classes.
Then by guarding the quantifications of these formulas, we simultaneously select the behavior classes of the three factors.
For example, if $\varphi$ is the $\FO$ formula describing a given class of $\sim_\mathcal{A}$, then we can ensure that $u[i,j]$ is of said class by  replacing each quantification $\exists z$ (resp. $\forall z$) in $\varphi$ by
$\exists z~x\leq z\wedge z\leq y$ (resp. $\forall z~x\leq z\wedge z\leq y$), after having renamed any occurrence of $x$ and $y$.

Finally, as there exists only a finite number of classes for each behavior relation, there exists only a finite number of combinations of classes for the three factors that satisfy the specification of $\varphi^{q,q'}(x,y)$, which can thus be written as the disjunction over every compatible triplet of classes for $u[1,i-1]$, $u[i,j]$ and $u[j+1,|u|]$.

\begin{figure}
\begin{center}

\begin{tikzpicture}

\draw[|-|] (0,6) -- (1,6);
\draw[-|] (1,6) -- (4,6);
\draw[-|] (4,6) -- (5,6);

\node (u) at (-0.7,6) {$u$};
\node (i) at (1,6.3) {$i$};
\node (j) at (4,6.3) {$j$};

\draw[dashed] (1,6) -- (1,3.7);
\draw[dashed] (4,6) -- (4,3.7);

\node (q) at (0.8,5.5) {$q$};
\node (qp) at (4.2, 4.1) {$q'$};

\draw (1,5.5) .. controls (3,5.5) and (3,5.1) .. (1,5.1);
\draw (1,5.1) .. controls (0,5.1) and (0,4.7) .. (1,4.7);
\draw (1,4.7) -- (4,4.7);
\draw (4,4.7) .. controls (5,4.7) and (5,4.4) .. (4,4.4);
\draw (4,4.4) -- (1,4.4);
\draw (1,4.4) .. controls (2,4.27) and (2,4.23) .. (1,4.1);
\draw (1,4.1) -- (4,4.1);

\end{tikzpicture}
\caption{A run between two positions is decomposed as a succession of partial runs, characterized by the behaviors of each factor.}\label{2w-FOT:FigRun}

\end{center}
\end{figure}
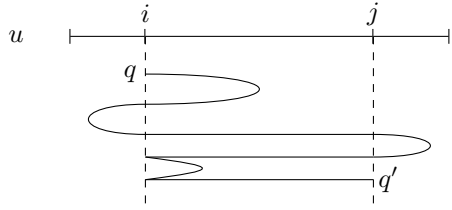
\end{proof}

\subsection*{From $\FO$ transductions to aperiodic two-way transducers}
\subsubsection*{Construction of the aperiodic transducer}
\setcounter{theorem}{17}
\begin{theorem}
Let $T$ be a $\FO$ graph transduction.
Then we can effectively construct an aperiodic two-way transducer with $\FO$ look ahead that realizes the same function over words.
\end{theorem}

\begin{proof}
Let $T=(A,B,\varphi_{dom},C,\varphi_{pos},\varphi_\leq)$ be a $\FO$ transduction.
Then we construct a two-way transducer with $\FO$ look around $\mathcal{A}=(Q,A,B,\Delta,i,\{f\})$ that realizes the function $T$, and prove that
it is aperiodic.
The main idea is that from the order formulas of $T$, we can deduce the successor relation. By setting the set of states as set of copies, the successor relation corresponds to the transitions of the transducer. We add an initial and a final states that are reached on the endmarkers.
Thus, we set $Q=C\uplus \{i,f\}$.

For readability purposes, we first set $\varphi_*^{c}(x)=\bigvee_{b\in B}\varphi_b^c(x)$, a formula which is satisfied if the $c$ copy of the node quantified by $x$ exists in the output structure. 
To define the transition relation, we now set, for each pair of copies, the successor formula
$$
 S^{c,c'}(x,y)= \varphi_\leq^{c,c'}(x,y)\wedge 		
 			\forall z \bigwedge_{d\in C}
					\big(\varphi_*^d(z)\to\bigl(\varphi_\leq^{d,c}							(z,x)\vee\varphi_\leq^{c',d}(y,z)\bigr)\big)
$$
Given a position $x$ and a copy $c$, there exists at most one $c'$ and one position $y$ such that $S^{c,c'}(x,y)$ is true, but only if the input word is in the domain of the transduction, and if we range over the existing nodes of the output structure.
Thus we define $\psi^{c,c'}(x,y)=S^{c,c'}(x,y)\wedge \varphi_*^c(x)\wedge\varphi_*^{c'}(y)\wedge \varphi_{dom}$.
The fact that the transducer can only move to existing nodes will ensure the determinism of the machine.
The transition relation $\Delta$ is then defined as set of tuples
$(c,\varphi_b^c(x),c',b,\psi^{c,c'}(x,y))$. Note that the transitions are mutually exclusive by definition of graph transduction, and that they also handle the production.
We also add to $\Delta$ some transitions regarding the endmarkers and the initial and the final states~:
\begin{itemize}
\item For the initial case, we define the formula 
$$first^c(y)=\varphi_*^c(y)\wedge\forall x \bigwedge_{d\in C}\big(\varphi_*^d(x)\to \varphi_\leq^{c,d}(y,x) \big)$$
that is satisfied by the first node of the output structure, and add the transition $(i,\vdash\!\!(x),c,\epsilon,first^c(y)\wedge\varphi_{dom})$ that moves the reading head from the initial position to the first node.
One should note that we consider without generating problems that the formula $first^c(y)$ is a formula $\psi(x,y)$ with two free variables, where $x$ is not used.
\item For the final case, we similarly define the formula
$$last^c(y)=\varphi_*^c(y)\wedge\forall x \bigwedge_{d\in C}\big(\varphi_*^d(x)\to \varphi_\leq^{d,c}(x,y)\big) $$
and add the transitions $(c,\varphi_b^c(x)\wedge last^c(x)\wedge\varphi_{dom},f,b,\dashv\!(y))$.
Note that these transitions handle the production of the label of the last node of $T$.
\end{itemize}
	
Since the production of $T$ over its domain is a linear graph, there is exactly one node that satisfies a formula $first^c(x)$ and exactly one that satisfies a formula $last^c(x)$ when the input satisfies $\varphi_{dom}$.
And since any other node has exactly one successor, the resulting transducer $\mathcal{A}$ is deterministic.

Now remark that the reading head of the transducer $\mathcal{A}$ follows exactly the output structure of $T$.
Then according to Lemma~\ref{FOT-2w:ParcFOT}, there exists an aperiodicity index $n$ such that for any words $u,v$ and $w$, if $vu^nw$ is in the domain of $T$, then $vu^{n+1}w$ is also in the domain of $T$ and the formulas of $T$ have the same truth value when the free variables range over $v$ and $w$.
Plus, as we consider first-order formulas, we know that we can choose an integer $n$ such that the words $u^n$ and $u^{n+1}$ satisfy the same formulas of $T$.
Then the moves of the input head on the words $vu^nw$ and $vu^{n+1}w$
either both are in iterations of $u$, or both are outside.
Then the context paths $path_{vw}(vu^nw)$ and $path_{vw}(vu^{n+1}w)$, which are the traces of the runs outside of $u^n$ and $u^{n+1}$ respectively, will be equal,
 proving the aperiodicity of $\mathcal{A}$.
\end{proof}

\end{document}